\documentclass{article}

\usepackage{arxiv}

\usepackage[utf8]{inputenc} %
\usepackage[T1]{fontenc}    %
\usepackage{hyperref}       %
\usepackage{url}            %
\usepackage{amsfonts}       %
\usepackage{nicefrac}       %
\usepackage{microtype}      %
\usepackage{graphicx}
\usepackage[numbers]{natbib}
\usepackage{doi}

\usepackage{amsthm}
\usepackage{algorithm}

\usepackage{tikz}
\usepackage{multirow}
\usepackage{amsmath}
\usepackage{array}
\usepackage{hyperref}
\usepackage{tabularx}
\usepackage{nicematrix}
\usepackage{subcaption}
\usepackage{adjustbox}
\usepackage{makecell}
\usepackage{todonotes}
\usepackage{bm}
\usepackage{arydshln}
\usepackage{booktabs}

\usepackage{amssymb}
\usepackage[most]{tcolorbox}

\usepackage{enumitem} %
\usepackage{algorithmic}

\usepackage[most]{tcolorbox}
\usepackage{amsthm}
\usepackage{thmtools}
\usepackage{thm-restate}

\newcommand{\X}{\mathbf{X}}

\def\namedlabel#1#2{\begingroup
    #2%
    \def\@currentlabel{#2}%
    \phantomsection\label{#1}\endgroup
} %

\usepackage{thmtools}
\usepackage{thm-restate}

\renewcommand{\X}{\mathbf{X}}

\newcommand{\Patch}{\mathbf{P}}

\newcommand{\FixPT}{\textbf{FixProp3}}
\newcommand{\FixPTX}{\textbf{FixProp3($\X$)}}

\newcommand{\FixPTM}{ \textbf{FixProp3Gen}}
\newcommand{\FixPTMX}{ \textbf{FixProp3Gen($\X$)}}

\newtheorem{lemma}{Lemma}[section]
\newtheorem{definition}[lemma]{Definition}

\newtheorem{observation}[lemma]{Observation}

\title{EFX Allocation in (Multi)Hypergraphs}

\author{ \href{}{ \hspace{1mm}Thanasis Lianeas} \\
	University of West Attica\\
	Greece \\
    \AND
	\href{https://orcid.org/0000-0003-3997-5131}{\includegraphics[scale=0.06]{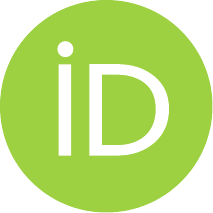}\hspace{1mm}Alkmini Sgouritsa} \\
	Athens University of Economics and Business\\
	Archimedes/Athena RC\\
	Greece \\
    \AND
	\href{https://orcid.org/0009-0007-3016-980X}{\includegraphics[scale=0.06]{orcid.pdf} \hspace{1mm}Minas Marios Sotiriou} \\
	Athens University of Economics and Business\\
    Archimedes/Athena RC\\
	Greece \\
}

\date{}

\renewcommand{\headeright}{}
\renewcommand{\undertitle}{}
\renewcommand{\shorttitle}{}

\hypersetup{
pdftitle={Title},
pdfsubject={},
pdfauthor={},
pdfkeywords={},
}

\begin{document}
\maketitle

\begin{abstract}
   We study fair allocations of {\em indivisible} goods among agents with heterogeneous monotone valuations. As fair we consider the allocations that are envy-free-up-to-any-good (EFX). Finding if EFX allocations always exist, even for agents with additive valuations, is a major open problem in Fair Division. 
Christodoulou et al. (2023) introduced the (multi-hyper)graph setting, where agents and goods are represented by vertices and edges of a graph, respectively, and only the endpoints of an edge may have non-zero marginal value for it. We show that for \textit{hypergraphs} with girth at least 4 and agents with \textit{general monotone} valuations there always exists an EFX allocation and can be constructed in polynomial time. We generalize our approach to also show that multi-hypergraphs with girth (on the simple hypergraph) at least 4 always admit an EFX allocation, as long as there exists a single vertex whose incident edges have multiplicity at most the size of that edge minus 2; our construction in this case needs pseudo-polynomial time.
\end{abstract}

\newpage

\section{Introduction}

Fair division of resources among agents is an important topic, from dividing inheritance, to allocating computational resources to training models. Many situations fall under the category of fair division of indivisible goods, which is the central theme of this paper. There are sites (e.g. http://www.spliddit.org/) that provide mechanisms for such applications, based on theoretical results. The research of fair division dates back to almost 80 years ago \cite{h__steihaus_1948}. A well studied and established notion of fair allocations are \textit{envy-free allocations}, where nobody envies another agent \cite{gamow1958puzzle, foley1966resource, VARIAN197463}, which always exists for divisible resources \cite{Stromquist1980HowTC, Woo80, Aziz2016}. However, envy-free allocations may not always exist regarding indivisible goods; consider for instance the case of two agents and a valued good, where whoever gets the good is envied by the other agent. Since the envy-free condition is too strict in this case, two natural relaxations were defined. The first notion is envy-freeness up to one good (EF1) \cite{budish}. Such an allocation always exists and can be computed in polynomial time \cite{LiptonEtAl}, even when agents have arbitrarily heterogeneous (monotone) valuations over the sets of goods. The second notion, which is the one considered here, is envy-freeness up to any good (EFX) \cite{EFXCara, GMT14} and it is stricter than EF1.  

In contrast to EF1, it is unknown if EFX allocations always exist, and this problem has been described as ``Fair Division's Most Enigmatic Question'' \cite{Procaccia}. An EFX allocation is known to exist only in special cases: for 2 agents with heterogeneous monotone valuations \cite{PlautRough}, for 3 agents with additive valuations or a slightly more general class of valuations \cite{CGM24, akrami2022efxallocationssimplificationsimprovements}, and for many agents with identical monotone valuations \cite{PlautRough}. Other results limit the valuation function of the agents or consider limited types of agents (e.g., \cite{TwoValuedInstanses, efxexistsfor3typesofadditiveagents, DBLP:conf/aaai/HosseiniSVX21, babaioff2020fairtruthfulmechanismsdichotomous}).

\citet{EFXsimplegraphs} introduced a restriction on the valuations by relying on a graphical structure. More precisely, agents and edges are represented by vertices and edges, respectively, and \emph{only} the agents that are the endpoints of an edge/good may consider it valuable. The motivation behind the graph setting is about instances where multiple agents find valuable a ``neighboring'' good but not all of them (e.g. geographic settings between neighboring countries, or allocation of work space between research labs). Moreover, the multi-hypergraph setting is basically the unrestricted setting, therefore, exploring ways of coping with graph settings may be proven useful for solving the general problem. \citet{EFXsimplegraphs} showed that it is not always possible to construct an EFX allocation by orienting the edges of the graph, i.e., by allocating each edge to one of its endpoints, even for 4 agents, however, they showed that EFX allocations always exist. 

There have been several follow-up works considering the existence of EFX allocations in graph settings: \citet{kaviani2024envyfreeallocationindivisiblegoods} showed existence of EFX allocations for the multigraph setting when agents have restricted additive valuations,\footnote{In restrictive additive valuations each agent values each good $g$ by either a fixed value $v_g$ for good $g$, or $0$. Then, the valuation for any set equals the sum of the values for each good in the set.} and there is a line of works showing existence of EFX allocations in multigraphs for more general valuations under restrictions on the graph structure 
\cite{afshinmehr2024efxallocationsorientationsbipartite, bhaskar2024efxallocationsmultigraphclasses, OntheExistenceofEFXAllocationsinMultigraphs}; a common restriction that appears in all those works is about the girth of the underlying simple graph. Our results is in the same direction and focus on the existence of EFX allocations in hypergraph and multi-hypergraph setting with girth at least 4. We note that for hypergraphs with girth at least 3, the best known result is only an approximation of the EFX guarantee ($\frac{\sqrt{2}}2-$EFX existence) for subadditive valuations \cite{kaviani2024envyfreeallocationindivisiblegoods}; we manage to show exact EFX guarantees for {\em general} monotone valuations by slightly relaxing the restriction on the girth.

\subsection{Our Results.}

We show that an EFX allocation always exists for hypergraphs of girth at least 4, when agents have general heterogeneous monotone valuations. We do this by constructing this allocation in polynomial time. The following theorem describes this first result.

\begin{restatable}{theorem}{thmSimple}
\label{thm:girth4}
    Instances on hypergraphs of girth at least 4 always admit an EFX allocation that can be constructed in polynomial time to the number of agents and goods.
\end{restatable}

We generalize our result, by considering a multi-hypergraph of girth at least 4: such a graph may contain multiple edges of the same subset of vertices, i.e., each edge may have multiplicity more than 1, but the simple underlying hypergraph (by considering multiplicity 1 for all edges) has girth at least 4. We remark that with no further restriction, the existence of EFX allocations in such a setting reduces to the general problem of EFX existence: consider a single edge containing all vertices of multiplicity equal to the number of goods. In our second result we prove the existence of EFX allocation in multi-hypergraphs of girth at least 4 as long as there exists a single vertex whose incident edges have multiplicity at most the size of that edge minus 2; our construction in this case needs pseudo-polynomial time. We justify the importance of this restriction by showing that if we slightly relax this restriction, the problem of EFX existence with at most one unallocated good reduces to the problem we try to solve; %
this problem is considered quite difficult and it is solved only for the case of 4 agents \cite{AlmostFullEFXforFourAgents}.  
Our second result is summarized in the  next theorem.

\begin{restatable}{theorem}{thmMulti}
\label{thm:girth4_multi}
    For multi-hypergraphs with girth at least 4 where there exists a vertex whose incident edges have multiplicity at most the edge's size minus 2, there exists an EFX allocation constructed in pseudo-polynomial time.
\end{restatable}

\paragraph{Our approach.}
It is well known that EFX orientations may not exist even in simple graphs (see the counterexample in \cite{EFXsimplegraphs} regarding the existence of EFX orientations), hence, it may be the case that in an EFX allocation some edges are allocated to vertices that are not endpoints of those edges. Our approach anticipates this fact and initially chooses an arbitrary vertex (vertex $0$ in our analysis) to serve as the vertex to ``park" those edges. One such vertex is sufficient in the case of simple hypergraphs (Theorem~\ref{thm:girth4}), however when we consider edges of higher multiplicity than 1 (Theorem~\ref{thm:girth4_multi}), more such vertices may be needed, and for that we employ the neighbors of vertex $0$ to play such a role.

\subsection{Further Related Work.}
\label{sec:furtherRW}

{\bf EFX in Multigraphs and hypergraphs}.
\citet{EFXsimplegraphs} introduced the graph setting and showed the existence of EFX allocation when the graph is simple and the items are goods. This result was extended to the case of mixed manna, where items may be both goods and chores \cite{mixedmanna}.   %
\citet{kaviani2024envyfreeallocationindivisiblegoods} showed that multigraphs where agents have restricted additive valuations always admit an EFX allocation. 
\citet{afshinmehr2024efxallocationsorientationsbipartite}, and \citet{bhaskar2024efxallocationsmultigraphclasses} proved that bipartite multigraphs where agents have additive and cancelable valuations, respectively, admit EFX allocations. In parallel, \citet{OntheExistenceofEFXAllocationsinMultigraphs} showed that multigraphs with girth at least 6, or multigraphs where each vertex is connected with roughly at most a quarter of the other vertices, admit EFX allocations for agents with general monotone valuations, and \citet{bhaskar2024efxallocationsmultigraphclasses}
showed that $t$ colored multigraphs with girth at least $2t-1$ for agents with cancelable valuations admit EFX allocations.  Subsequently, \citet{afshinmehr2025efxallocationsexisttrianglefree} showed that triangle-free multigraphs (girth $\ge 4$) always admit EFX allocations for general monotone valuations. Recently, in parallel two papers in different approaches showed existence of EFX allocation on multigraphs with cancelable valuations \cite{EFXmultigraphsChristodoulou, afshinmehrMultigraphs2026efx}.

Approximate EFX-allocations have also been studied in multigraphs and hypergraphs. \citet{amanatidis2024pushingfrontierapproximateefx} showed that $\frac{2}{3}-$EFX allocations always exist in multigraphs when agents have additive valuations, which was recently improved to a $\frac{\sqrt{2}}{2}-$EFX \cite{kaviani2025improvedapproximateefxguarantees}. \citet{kaviani2024envyfreeallocationindivisiblegoods} showed that for hypergraphs with 
 girth at least 3 (meaning that any two agents share at most one edge), $\frac{\sqrt{2}}{2}$-EFX allocations always exist when agents have  subadditive valuations.  

In the graph and multigraph setting the existence of EFX orientations, i.e., allocations where edges may only be allocated to one of the endpoints, has also been considered. \citet{EFXsimplegraphs} showed that EFX orientations need not exist by giving a counterexample in a $K_4$ graph, and they further showed that even deciding if an EFX orientation exists  is NP-complete; it was later shown that this result holds even if the vertex cover of the graph has size  8, or in multigraphs with only 10 vertices \cite{deligkas2024ef1efxorientations}, which was later improved to a vertex cover of size 4 or multigraphs with as few as 4 vertices by~\cite{kanellopoulos2025efxorientationsparameterizedcomplexity}. 
\citet{ZengStructureOrientationsGraphsAAMAS} showed that EFX orientations may not 
exist in graphs with chromatic number greater than 3, and they always exist when the chromatic number is at most 2. The complexity of orientations  has been further explored, e.g.,  \cite{hsu2024efxorientationsmultigraphs,afshinmehr2024efxallocationsorientationsbipartite,TractableGraphStructures}.

Another fairness criterion known as maximin share (MMS) has also been explored in multigraphs \cite{christodoulou2025exactapproximatemaximinshare,feige2025multiallocationsallocationssubadditivevaluations}.

{\bf EFX with charity}.
The concept of (partial) EFX allocation with unallocated goods, known also as EFX allocation with charity, was introduced by \citet{CGH19}. \citet{CKMS21} showed that EFX allocations always exist, even with general monotone valuations, if at most $n-1$ goods are donated to charity, where $n$ is the number of agents, and moreover nobody envies the charity. The size of the charity was then improved to $n-2$, and to one for the case of 4 agents with  valuations slightly more general than additive \cite{AlmostFullEFXforFourAgents}. For hypergraphs with girth at least 3, the size of the charity was reduced to $\lfloor \frac{n}{2} \rfloor -1$ for general monotone valuations \cite{kaviani2024envyfreeallocationindivisiblegoods}. Finally, the number of unallocated goods was subsequently improved to sublinear by \cite{CGMMM21, akrami2022efxallocationssimplificationsimprovements, berendsohn2022fixedpointcyclesefxallocations, jahan2023rainbowcyclenumberefx}, but for approximate EFX.

\section{Preliminaries}
\label{sec:prel}
We consider a setting where there is a set $N$ of $n$ agents and a set $M$ of $m$ indivisible goods, and each agent $i$ has a valuation function $v_i: 2^M \rightarrow \mathbb{R}_{\geq 0}$, over the subsets of goods, i.e., $v_i(S)$  denotes the valuation of agent $i$ for the subset $S$ of goods. The valuation functions are considered to be monotone, i.e.,  for any $S\subseteq T \subseteq M$, it holds that $v_i(S) \leq v_i(T)$, and normalized, i.e., $v_i(\emptyset)=0$. For simplicity, for the valuation of $i$ for some good $g$, we write $v_i(g)$ instead of $v_i(\{g\})$. 
An {\em allocation} $\X=(X_1, \ldots, X_n)$ is a partition of a subset of goods into $n$ disjoint bundles $X_1, \ldots, X_n$, where each agent $i$ receives $X_i$. An allocation is \emph{complete} if it allocates all the goods, and \emph{partial} if not. 
Throughout, for $k\in \mathbb{N}^+$, we let $[k]:=\{1,2,\ldots,k\}$.  

{\bf Envy - EFX allocation.}
    Given an allocation $\X=(X_1,\ldots, X_n)$, an agent $i$ {\em envies} an agent $j$ (or alternatevly $X_j$), if
    $v_i(X_i) < v_i(X_j)$. 
    An allocation $\X=(X_1,\ldots, X_n)$ is {\em EFX} if for any pair of agents  $i,j$ it holds that:
    $$v_i(X_i) \geq v_i(X_j \setminus \{g\}) \hspace{1ex} ,\forall g \in X_j.$$
    In other words, in an EFX allocation, no agent envies a proper subset of what is allocated to any other agent.

We assume that the setting is modeled on hypergraphs. 
A \emph{hypergraph} is a pair $G=(V,E)$ where $V$ is a set of vertices and $E$ is a set of subsets of $V$, called hyperedges or simply edges. The size of an edge is the number of vertices it contains. For an edge $e\in E$ and a vertex $i\in e$ we say that $e$ is incident to $i$ and $i$ is incident to $e$, or $i$ is an endpoint of $e$. %
If vertices  $i$ and $j$ belong to some edge $e$ (i.e., $i,j\in e$), we say that $i$ and $j$ share $e$ and we call $i$ and $j$ neighbors.

{\bf Cycles \cite{berge1973graphs} - Girth.}
    For a hypergraph  %
    $G$, a cycle of length k is defined by a sequence $(x_1,e_1,x_2,e_2, \dots ,x_k,e_k,x_{k+1})$ such that:
           (i) $x_1,x_2, \dots, x_k$ are distinct vertices and $x_{k+1} = x_{1}$,
        (ii) $e_1,e_2, \dots, e_k$ are  distinct edges and (iii)
         $x_i,x_{i+1} \in e_i$, $\forall i \in [k]$.
    A hypergraph $G$ has \emph{girth} $k$, if $G$'s shortest cycle has length  $k$; if there is no cycle in $G$ the girth is infinity.

{\bf Hypergraph setting.} 
In the hypergraph setting,  we let agents correspond to vertices of a hypergraph $G=(V,E)$, and goods correspond to  edges of $G$. For ease, we  refer to the agents as vertices and to goods as edges. The edges that are not incident to some vertex $i$ are {\em irrelevant} to it, i.e., for any $S\subseteq E$ and any $e \in E$ such that $i \notin e$, $v_i(S \cup \{e\}) = v_i(S)$. Note that this implies that for $e\in E$ irrelevant  to $i\in V$: $v_i(e)=v_i(\emptyset)=0$. We call an edge $e$ {\em relevant} to $i$, if $i \in e$.

{\bf Orientation.}
An allocation $\X$ is an {\em orientation} if for any allocated edge $e$, if $e \in X_i$, then $i$ is an endpoint of $e$. %

{\bf Unallocated edges.}
    Given a partial allocation $\X$, an edge $e$ is \emph{unallocated} if $e\notin X_i$, for all $i\in V$. We denote by $U(\X)$ the set of unallocated edges in $\X$. For each vertex $i$, we define $U_i(\X)$ to be the set of all the unallocated edges that are relevant to $i$.

We give the following two observations for  hypergraphs with girth at least $3$ and at least $4$. For hypergraphs with girth at least $3$, any two vertices may share at most one edge, and for hypergraphs with girth at least $4$, any two vertices of an edge $e$ do not share any other common neighbor outside $e$.

\begin{observation}\label{obs:at-most-1-relevant-common-edge}
    In hypergraphs with girth at least 3, %
    any two vertices may share at most one edge.
\end{observation}
\begin{proof}
    On the contrary, let $i,j$ be two vertices that both belong to  edges, say,  $e_1$ and $e_2$. Starting with vertex $i$, following edge $e_1$ to reach $j$ and then following edge $e_2$ to reach $i$ we get a cycle of length 2, i.e., the cycle $(i,e_1,j,e_2,i)$, contradicting the hypothesis that the girth is at least 3.
\end{proof}

\begin{observation}\label{obs:no_Triangles}
    In hypergraphs with girth at least 4, any two vertices of an edge $e$ do not share any other common neighbor outside $e$.
\end{observation}

\begin{proof}
    On the contrary, let vertices $x,y$ belong in some edge $e$ and let $v$ be one of their common neighbors  outside $e$. Let $e_x$ be an edge that $x$ and $v$ share. By Observation \ref{obs:at-most-1-relevant-common-edge} $y$ cannot belong in $e_x$, or else $x$ and $y$ would share $e$ and $e_x$. Let $e_y$ be an edge that $y$ and $v$ share. The cycle $(x,e_x, v,e_y,y,e,x)$ has length 3, which is a contradiction.  
\end{proof}

{\bf Multi-hypergraph setting.} In Section \ref{sec:multi-hyper} we allow for a more general hypergraph setting. In the \emph{multi-hypergraph setting}  
the edge set $E$ is allowed to have edges of the same set of vertices. Such edges correspond to different goods, with possibly  different impact on the valuation functions of the vertices/agents.
For an edge $e$ that appears $k$ times in $E$ we say that $e$ has multiplicity $k$.
For a multi-hypergraph $G=(V,E)$ we define the girth of $G$ to be the girth of $G'=(V,E')$, where $E'$ is derived from $E$ if we delete all repetitions of the  edges of $E$.

\section{EFX on Hypergraphs}\label{section:EFXhyper}

For ease of presentation we will rename/reorder the vertices using numbers from $0$ to $n-1$.  Pick an arbitrary vertex and let it be vertex $0$. 
Let $n_0$ be the number of neighbors that vertex $0$ has. %
 Arbitrarily, name the neighbors of $0$ with numbers from $1$ to $n_0$ and the remaining vertices  with numbers from $n_0+1$ to $n-1$.
Throughout the runs of the presented algorithms, vertices will be prioritized based on their names-labels, in a decreasing order.

We next repeat the main theorem of the section whose proof is built up by the use of several lemmas. 

\thmSimple*

The proof is constructive and is described in the following subsections. At each time we preserve an EFX allocation by satisfying additional properties. Our intermediate goal, achieved by Algorithm \ref{Algo:Algo2}, is to achieve a partial allocation by orienting edges that satisfies the following four Properties. 

\begin{tcolorbox}[colback=black!5!white,colframe=black!75!black]
\begin{enumerate}
    \item \label{propSec3:EFXorientation} $\X$ is a (partial) EFX orientation.
    \item \label{propSec3:Non-EnviedNeighbors} Vertex $0$ is non-envied and any neighbor $i$ of vertex $0$ may  be envied only if it is allocated the edge that it shares with $0$. %
        \item \label{propSec3:No-envyToSingleEdges} For any vertex $i$ and $e \in U_i(\X), v_i(X_i) \geq v_i(e)$.
    \item \label{propSec3:No-envyToFlowers} For any envied vertex $i$, $v_i(X_i) \geq v_i\left(U_i(\X)\right)$. 
\end{enumerate}
\end{tcolorbox}

Properties \eqref{propSec3:EFXorientation}-\eqref{propSec3:Non-EnviedNeighbors} are preserved throughout Algorithm~\ref{Algo:Algo2}, and Algorithm \FixPT\ therein guarantees Property \eqref{propSec3:No-envyToSingleEdges}. At the end of Algorithm \ref{Algo:Algo2}, a stronger property for the envied vertices than \eqref{propSec3:No-envyToSingleEdges} is satisfied, namely Property \eqref{propSec3:No-envyToFlowers}.  
In Algorithm \ref{Algo:FinalAllocation} and Lemma~\ref{lem:final_alloc} we show how those properties are used in order to construct a complete EFX allocation.

{\bf Proof of Theorem~\ref{thm:girth4}.} For the proof of Theorem~\ref{thm:girth4} it suffices to run in series Algorithms \ref{Algo:Algo2} and \ref{Algo:FinalAllocation}. By Lemmas \ref{lem:prop124}, \ref{lem:0non-envied}, \ref{lem:Prop3_2ndPart} and \ref{lem:final_alloc} we show that if the instance is a hypergraph of girth at least 4, the final allocation will be a complete EFX allocation. In Lemma~\ref{lem:complexity} we show that those algorithms run in polynomial time, which completes the proof. %

\subsection{Orienting edges.}
\label{sec:OrientEFX}
Algorithm~\FixPT~is called in order to ensure that no unallocated edge is preferred by any vertex over its bundle.

\begin{algorithm}[t]
\caption{ \FixPT}
\label{Algo:Round Robin}
\raggedright\textbf{Input:} An allocation $\X$  satisfying Property \eqref{propSec3:EFXorientation} \\
\textbf{Output:} An allocation $\X$  satisfying Properties \eqref{propSec3:EFXorientation} and \eqref{propSec3:No-envyToSingleEdges}.\\
\begin{algorithmic}[1]
\WHILE{$\exists$ $j \in V, e \in U(\X)$: $v_j(e) > v_j(X_j)$}
  \STATE Let j be the maximum vertex with this property for some $e$
   \STATE $X_j \gets \{\arg\max_{e \in U_j(\X)} v_j(e)\}$
\ENDWHILE 
\RETURN $\X$

\end{algorithmic}
\end{algorithm}

\begin{lemma}\label{lem:procProp}
    When \FixPT\ is called for an EFX orientation, it outputs an EFX orientation satisfying Property \eqref{propSec3:No-envyToSingleEdges}.
\end{lemma}

\begin{proof}
First note that \FixPT\ will terminate since the bundles (and thus the values) that a vertex may get are finite, and at every execution of the while-loop, the value of some vertex for the updated allocation strictly increases. 

To see that Property \eqref{propSec3:No-envyToSingleEdges} is satisfied in the resulting allocation, see that the condition of the while-loop is the negation of Property \eqref{propSec3:No-envyToSingleEdges}, so \FixPT\ terminates  by satisfying it.

Clearly, the allocating step of line 3  orients a single (unallocated) edge to one of its incident vertices. This keeps the allocation an orientation. This allocation further remains an EFX allocation after each execution of the while-loop, since for the new bundle/edge allocated in line 3, removing the single edge it contains will leave it empty and thus non-envied. Also the value of the vertices may only increase during the execution of \FixPT; thus, no further envy may appear in the future and the EFX property will not break.
\end{proof}

Algorithm \ref{Algo:Algo2} starts by assigning each vertex, in decreasing order, its most valuable unallocated edge (one call of \FixPT). It continues by repeatedly offering envied vertices all their incident unallocated edges in place of their current bundle. If an envied vertex prefers all its incident unallocated edges to its currently allocated bundle, it is allocated those edges, it releases its bundle, followed by one call of \FixPT, i.e., repeatedly, any unallocated single edge is offered to vertices, in decreasing order, until no vertex prefers an unallocated edge. This algorithm is similar to  Algorithm 2 in \cite{EFXsimplegraphs}, but edges might have size greater that 2, and an order is added when offering a single unallocated edge, to make sure that Property \eqref{propSec3:Non-EnviedNeighbors} is satisfied.

\begin{algorithm}[t]
\caption{Orienting Edges}
\label{Algo:Algo2}
\raggedright\textbf{Input:} A hypergraph $G$ of girth at least 4.\\
\textbf{Output:} An allocation $\X$ satisfying Properties \eqref{propSec3:EFXorientation}-\eqref{propSec3:No-envyToFlowers}.\\
\begin{algorithmic}[1]
\STATE Let $\X$ be the all empty allocation.
\STATE \FixPTX
\WHILE{ $\exists$ envied $i\in V$: $v_i(U_i(\X)) > v_i(X_i)$}
\STATE $X_i \gets U_i(\X)$
\STATE \FixPTX  
\ENDWHILE
\end{algorithmic}
\end{algorithm}

\begin{lemma}
    \label{lem:prop124}
    After the termination of  Algorithm \ref{Algo:Algo2}, Properties \eqref{propSec3:EFXorientation}, \eqref{propSec3:No-envyToSingleEdges} and \eqref{propSec3:No-envyToFlowers} are satisfied.
\end{lemma}

\begin{proof}
Algorithm \ref{Algo:Algo2} will terminate since both \FixPT{ }and the allocating step of line 4 only strictly increase the value that a vertex gets and the bundles (and thus the values) that a vertex may get are finite. The condition of the while-loop of Algorithm \ref{Algo:Algo2} is the negation of Property \eqref{propSec3:No-envyToFlowers}. Since the condition of the while-loop must be false %
for the algorithm to terminate, %
Property \eqref{propSec3:No-envyToFlowers} will hold. Additionally, since the algorithm ends with a call of \FixPT, by Lemma \ref{lem:procProp}, Properties \eqref{propSec3:EFXorientation} and \eqref{propSec3:No-envyToSingleEdges} will also be satisfied as long as Property \eqref{propSec3:EFXorientation} was satisfied before the call of \FixPT. 

It remains to show that Property \eqref{propSec3:EFXorientation} was satisfied before any call of \FixPT. The empty allocation clearly satisfies Property \eqref{propSec3:EFXorientation}, so it is satisfied before the call of \FixPT\ at line 2. Suppose that Properties \eqref{propSec3:EFXorientation} and \eqref{propSec3:No-envyToSingleEdges} are satisfied before an execution of any round of the  while-loop, which is the case before the first execution of the while-loop. It suffices to show that the allocating step of line 4 does not break the EFX property (it clearly gives an orientation); this would guarantee that Property \eqref{propSec3:EFXorientation} is satisfied before \FixPT, and by Lemma~\ref{lem:procProp}, Properties \eqref{propSec3:EFXorientation} and \eqref{propSec3:No-envyToSingleEdges} would be indeed satisfied before the next round of the while-loop. Consider a vertex $i$ that  will get its bundle changed by the allocating step of line 4. By Observation \ref{obs:at-most-1-relevant-common-edge}, for any other vertex $j$, $U_i(\X)$ contains at most one edge, say $e$, relevant to $j$. By Property \eqref{propSec3:No-envyToSingleEdges}, vertex $j$ does not prefer $e$, if it exists, i.e., $v_j(X_j)\geq v_j(e)=v_j(U_i(\X))$ and thus $j$ will not envy $i$. Also the value of the vertices may only increase during the execution of Algorithm~\ref{Algo:Algo2}, and therefore no further envy may appear and the EFX property will not break.
\end{proof}

We will prove that Algorithm \ref{Algo:Algo2} will output an allocation also satisfying Property \eqref{propSec3:Non-EnviedNeighbors} in the next two lemmas.

\begin{lemma}
    \label{lem:0non-envied}
    Vertex $0$ is non-envied throughout the execution of Algorithm \ref{Algo:Algo2}.
\end{lemma}
    \begin{proof} Initially it is $X_0=\emptyset$ and thus $0$ is non-envied. Any time that \FixPT\ is called, in order for $0$ to be considered for a change in its bundle, it must be that all other vertices do not prefer any of the unallocated edges, since all of them have higher index from $0$. This directly implies that if $0$ is allocated some edge by \FixPT, no other vertex will envy it and thus, since the values of the vertices may only increase, $0$ remains non-envied. Hence, vertex $0$ is never considered in the while-loop of Algorithm~\ref{Algo:Algo2}, and overall, it remains non-envied throughout the execution of Algorithm \ref{Algo:Algo2}.
\end{proof}

\begin{lemma}\label{lem:Prop3_2ndPart}
    After the termination of Algorithm \ref{Algo:Algo2}, any neighbor $i$ of vertex $0$ may only be envied if it is allocated the edge that it shares with $0$.
\end{lemma}

\begin{proof}
Consider any vertex $i \in [n_0]$, i.e., a neighbor of $0$, and the bundle $X_i$ allocated to $i$ after the termination of Algorithm \ref{Algo:Algo2}. If $|X_i|=0$, obviously $i$ is non-envied. If $|X_i|>1$, then $i$ is still non-envied due to Property \eqref{propSec3:EFXorientation}: Consider any other vertex $j$ and let $e$ be the edge, if any, that it shares with $i$; by Observation~\ref{obs:at-most-1-relevant-common-edge}, $i$ and $j$ may share at most one edge. If $e\notin X_i$, obviously $v_j(X_i)=0$, and $j$ doesn't envy $i$. If $e\in X_i$, since $|X_i|>1$, let $e'\neq e$ be some other edge in $X_i$. Then, by Property \eqref{propSec3:EFXorientation} it should be that $j$ does not prefer $X_i\setminus\{e'\} \supseteq \{e\}$ to the bundle allocated to $j$. Since $e$ is the only valuable edge for $j$ in $X_i$, $j$ does not envy $i$.  

In the case that $|X_i|=1$, let $e$ be the edge assigned to $i$ and suppose that $e$ is not the edge that $i$ and $0$ share. Since Algorithm \ref{Algo:Algo2} only orients edges, $e$ is relevant to $i$. We first show that $e$ cannot be relevant to any other neighbor of $0$, and therefore none of them envy $i$. Let $j \in [n_0]$ be some neighbor of $0$ belonging to the edge that $0$ and $i$ share. By Observation~\ref{obs:at-most-1-relevant-common-edge}, $e$ is irrelevant to $j$. Let $j \in [n_0]$ be some neighbor of $0$ not belonging to the edge that $0$ and $i$ share. 
Due to Observation \ref{obs:no_Triangles},  vertices $i,j$ cannot both belong in the same edge, say $e'$, or else they would have vertex $0$ as a common neighbor outside $e'$. Thus, $e$ is irrelevant to $j$. 

Vertices indexed higher than $n_0$ did not envy $i$ at the time it received $X_i$ during the execution of Algorithm~\ref{Algo:Algo2}: If $X_i=\{e\}$ was allocated to $i$ in a run of \FixPT, 
vertex $i$ was considered as the highest indexed vertex that preferred $e$ to what it got allocated, implying that all vertices with index higher than $i$ preferred their bundles over every single unallocated edge, so $e$ as well. If $X_i$ was allocated to $i$ in the while-loop of Algorithm~\ref{Algo:Algo2}, then Property \eqref{propSec3:No-envyToSingleEdges} (from the previous \FixPT\ run) guarantees that nobody envied $i$ at that point. Note that $i$ remains non-envied until the termination of Algorithm~\ref{Algo:Algo2}, since the vertices' value may only increase.
\end{proof}

\subsection{Complete EFX allocation.}
\label{sec:step3EFX}

In Algorithm~\ref{Algo:FinalAllocation} we allocate in two steps the remaining edges (if any) to reach a complete EFX allocation. Unallocated edges incident to a non-envied vertex are oriented towards some non-envied vertex and the rest are allocated to $0$.

\begin{algorithm}[t]
\caption{Complete allocation}
\label{Algo:FinalAllocation}
\raggedright\textbf{Input:} The allocation $\X$ returned by Algorithm~\ref{Algo:Algo2}.\\
\textbf{Output:} A complete EFX allocation $\X$.\\
\begin{algorithmic}[1]
\WHILE{ $\exists e\in U(\X)$ containing a non-envied vertex $j$}
        \STATE $X_j \gets X_j \cup \{e\}$
\ENDWHILE
\STATE $X_0 \gets X_0 \cup U(\X)$
\end{algorithmic}
\end{algorithm}

\begin{lemma}\label{lem:final_alloc}
    If $\X$ is  the allocation returned by Algorithm~\ref{Algo:Algo2}, then Algorithm~\ref{Algo:FinalAllocation} returns a complete EFX allocation.
\end{lemma}

\begin{proof}
There are two steps in Algorithm~\ref{Algo:FinalAllocation} for allocating the remaining edges, one in line 2, and the other in line 4. We will show that at each of them no further envy is created. In line 2, each unallocated edge containing a non-envied vertex is oriented towards one of its non-envied vertices. This keeps the allocation EFX. To see this, first note that when an edge $e$ is allocated to a vertex $i$ in this way, only vertices incident to $e$ may envy $i$, and $e$ will be the only edge that they share with $i$ (Observation \ref{obs:at-most-1-relevant-common-edge}). Since $\X$ is an orientation, this implies that for any of these vertices, say $j$, $v_j(e)=v_j(X_i\cup \{e\})$. Yet, due to Property \eqref{propSec3:No-envyToSingleEdges}, $j$ prefers its bundle over $e$, i.e., $v_j(X_j)\geq v_j(e)=v_j(X_i\cup e)$, and thus it will not envy $i$.

Note that at the end of the while-loop, the allocation is still an orientation and Properties \eqref{propSec3:EFXorientation}-\eqref{propSec3:No-envyToFlowers} are still satisfied. Additionally, the remaining unallocated edges have all their incident vertices envied. For the rest of the proof, $\X$ represents the updated allocation after the while-loop.

 It remains to show that allocating the rest of the unallocated edges to $0$ will not create any envy towards $0$. Since all those edges are relevant only to envied vertices, we will show that any envied vertex $i$ will not envy vertex $0$ after the allocation in line 4. If $i$ is a neighbor of vertex $0$, due to Property \eqref{propSec3:Non-EnviedNeighbors}, $i$ has received the shared edge with $0$, and therefore $i$ has no value for $X_0$. The same holds if $i$ is not a neighbor of vertex $0$, since after the while-loop, the allocation is still an orientation. Therefore, in both cases, due to Property \eqref{propSec3:No-envyToFlowers}, $v_i(X_i)\ge v_i(U_i(\X))=v_i(X_0 \cup U(\X))$, and so $i$ does not envy vertex $0$ in the final allocation.  
\end{proof}

\subsection{Poly-time EFX construction.}
\label{sec:Complexity}
To complete the proof of Theorem~\ref{thm:girth4}, we show that the construction of the EFX allocation needs polynomial time.

\begin{lemma}
\label{lem:complexity}
The construction of the EFX allocation by running  Algorithms \ref{Algo:Algo2}, and \ref{Algo:FinalAllocation} needs polynomial time complexity on the number of edges and vertices.
\end{lemma}
\begin{proof}
    Algorithm~\ref{Algo:Algo2} uses \FixPT\ as a subroutine, which runs in time $O(n^4)$: Each while-loop needs at most $n^2$ checks to find an appropriate vertex (since each vertex has degree at most $n-1$), and if those checks follow the priority of the vertices, no extra time is needed to find the maximum vertex (at line 2). Each vertex may update its allocation, i.e., be considered in the while, at most $n$ times (it has at most $n$ relevant edges and any time it strictly increases its value). So, overall the while-loop may be executed at most $n^2$ times and each execution needs $O(n^2)$ time. 
    
    Regarding Algorithm~\ref{Algo:Algo2}, the vertex picked in the while is changed from envied to non-envied (in line 4). We remark that in \FixPT , a non-envied vertex may turn to an envied one, however the value each vertex has for its allocated bundle always strictly increases when it updates her bundle. Therefore, each vertex may turn from non-envied to envied at most $n$ times, since every envied vertex receives a single edge and $n$ is an upper bound of its degree. Therefore, the while-loop may be executed at most $n^2$ times. Overall, the time complexity of Algorithm~\ref{Algo:Algo2} is $O(n^6)$.
    
    Algorithm~\ref{Algo:FinalAllocation} allocates at most $m$ edges, each in time $O(n)$, which is the time needed to identify if there exists a non-envied endpoint. Therefore, the time complexity of Algorithm~\ref{Algo:FinalAllocation} is $O(nm)$. 
    Hence, the construction of an EFX allocation needs overall polynomial time on $n$ and $m$. 
\end{proof}

\section{EFX on Multi-hypergraphs}
\label{sec:multi-hyper}

In this section we generalize our approach so it can be applied to the more general setting of multi-hypergraphs. We next repeat the main theorem of this section whose proof is built up by the use of several lemmas. 

\thmMulti*

Before proceeding to the proof of Theorem~\ref{thm:girth4_multi} we show the necessity of the additional restriction, in the sense that dropping it makes our problem at least as hard as a difficult problem in the literature.
More precisely, we construct an instance of a multi-hypergraph with girth at least 4, where there is no vertex whose all incident edges have multiplicity at most the edge’s size minus 2, and there exists a vertex that violates this condition only for one of its incident edges, for which the multiplicity is its size minus 1. We show that the general problem of EFX existence with at most one unallocated good reduces to finding an EFX allocation in that instance. We stress out that the problem of EFX existence with at most one unallocated good is considered a hard problem (see Section~\ref{sec:furtherRW}). 
\begin{lemma}
    Consider any instance $\mathcal{I}$ of $n$ agents and $m\geq n$ goods, where agents have arbitrary monotone positive valuations. Then, there exists an instance $\mathcal{I}'$ on a multi-hypergraph with   $n+1$ vertices, $m+1$ edges and girth at least 4, where there is no vertex whose all incident edges have multiplicity at most that edge’s size minus 2, and there exists a vertex with a single incident edge of multiplicity equal its size minus 1, and an EFX allocation in $\mathcal{I}'$ implies an EFX allocation with at most one unallocated good in $\mathcal{I}$. 
\end{lemma}

\begin{proof}
    Let $N=[n]=\{1,\ldots,n\}$ and $M=\{e_1, \ldots e_m\}$ be the set of agents and goods, respectively, in $\mathcal{I}$. We construct a multi-hypergraph instance $\mathcal{I}'$ with vertex set $N'=N\cup\{0\}$, and edge set $M'=M\cup\{e_0\}$, where $e_0=\{0,1\}$, and $e_j=[n]$ for all $e_j\in M$, i.e., each edge apart from $e_0$ contains all the vertices but $0$.
   Note that in this example there is no cycle (the girth is infinity) and there is no vertex whose incident edges have multiplicity at most that edge’s size minus 2 and vertex $0$ has a single incident edge of multiplicity equal its size minus 1. For any vertex $i\neq 0$, the valuation function $v_i$
over $M$, coincide with the valuation that they have in $\mathcal{I}$. 
The valuation for vertex $1$ for any set $S$ containing $ e_0$ is $0$, if $S=\{e_0\}$, and greater than $v_1(M)$ (its value for all edges except $e_0$), otherwise. Vertex $0$ has some positive value for $e_0$. 

We argue that in any EFX allocation in $\mathcal{I}'$, vertex $0$ is allocated $e_0$ and at most one other edge. First note that if \emph{only} $e_0$ was allocated to any vertex $i\neq 0$, $i$ has value $0$, and some vertex $j\geq 1$ would be allocated at least 2 edges from $M$. Since vertex $i$ values positively any of those edges, the EFX condition would be violated for $i$ against $j$. Therefore, if $e_0$ was allocated to some vertex $i\neq 0$ in any EFX allocation, vertex $i$ should receive at least one more edge. In that case, vertex $0$ would envy vertex $i$ after the removal of that edge (since vertex $0$ values positively only $e_0$), and the EFX condition would again be violated. So, we have established that in any EFX allocation, $e_0$ should be allocated to vertex $0$. If vertex $0$ was allocated at least two more edges, after the removal of any of those, vertex $1$ would still envy vertex $0$, which is again a violation of the EFX condition. 

Hence, overall, in any EFX allocation vertex $0$ is allocated $e_0$ and at most one other edge. In any such EFX allocation, the EFX condition is satisfied between the vertices in $N$, and their allocation gives an EFX allocation of the instance $\mathcal{I}$ with at most one unallocated good, namely the good that is possibly given to vertex $0$ apart from $e_0$.
\end{proof}

We now proceed to the proof of Theorem \ref{thm:girth4_multi}. The proof is constructive and generalizes the EFX construction for simple hypergraphs (Section~\ref{section:EFXhyper}). Similarly, our intermediate goal here is the construction of a (partial) EFX allocation that satisfies four properties. The first two properties are as in Section~\ref{section:EFXhyper} and the other are more generalized to handle the allowance of repetitions of edges. For that we  need the following definition regarding multiple appearances of edges.

\begin{definition}
    For a set $E$ of edges where repetitions are allowed,  $\Patch_e$ is the set of edges containing all the appearances of $e$ in $E$ (i.e., edges with the same incident vertices); we refer to $\Patch_e$ as the \emph{patch} for $e$.
    If some edges of $E$ are allocated by $\X$, $U_e(\X)$ denotes the subset of $\Patch_e$ that contains all the  unallocated edges of $\Patch_e$ under $\X$, i.e. $U_e(\X)=U(\X)\cap \Patch_e$. %
\end{definition}

Below we state the four properties for the multi-hypergraph setting. Property \eqref{propSec4:No-envytoPatches} differs from the corresponding property in simple hypergraphs in expressing no envy towards the whole unallocated set of any patch; in simple hypergraphs this was just a single edge. Property \eqref{propSec4:No-envyToFlowers} is extended to consider all vertices (apart from the special vertex $0$) and not only the envied ones as in the case of simple hypergraphs; the reason is that we may not be able to orient all edges incident to non-envied vertices, so we need to guarantee that EFX doesn't break when those are allocated to non-incident vertices.\footnote{We remark that we could transform Algorithm~\ref{Algo:Algo2} for simple hypergraphs to satisfy this more extended Property \eqref{propSec4:No-envyToFlowers}, however this was not necessary and so we kept only the necessary steps.} The four properties are used in Algorithm \ref{Algo:FinalAllocationMultipleEdges} and Lemma~\ref{lemma:FinalAllocationMulti} to construct a complete EFX allocation.

\begin{tcolorbox}[colback=black!5!white,colframe=black!75!black]
\begin{enumerate}
    \item\label{propSec4:EFXorientation} $\X$ is a (partial) EFX orientation.
   
    \item\label{propSec4:Non-EnviedNeighbors} Vertex $0$ is non-envied and any neighbor $i$ of vertex $0$ may  be envied only if it is allocated some edge(s) that it shares with $0$.

    \item\label{propSec4:No-envytoPatches} For any vertex $i$  and any $e \in U_i(\X)$, $v_i(X_i) \geq v_i(U_e(\X))$.

    \item\label{propSec4:No-envyToFlowers} For any vertex $i \neq 0$, $v_i(X_i) \geq v_i\left(U_i(\X)\right)$. %
    
\end{enumerate}
\end{tcolorbox}

Before we delve deep into the proof of Theorem~\ref{thm:girth4_multi} we give a high-level proof sketch.

\paragraph{Proof sketch.}

We keep a similar approach to Section~\ref{section:EFXhyper}: we employ an algorithm that outputs an allocation satisfying the four properties, and an algorithm that completes the EFX allocation. We define the vertices' labels as in Section~\ref{section:EFXhyper} by setting the special vertex whose incident edges  have the restricted multiplicity as vertex $0$.

Algorithm~\ref{Algo:OrientationForMultiplicityGreaterThan2} that satisfies the four properties is the same with Algorithm~\ref{Algo:Algo2}, by substituting the Algorithm \FixPT~with the more generalized Algorithm \FixPTM. \FixPTM\ basically applies rule $U_1$ of \cite{CKMS21} to each patch separately: It checks if there is a set $Z$ of unallocated edges in a \emph{single} patch such that there exists a vertex $k$ that values it more than its bundle, 
while for the other vertices the EFX condition is satisfied if $k$ receives $Z$, i.e., they do not envy any proper subset of $Z$. As long as such a vertex $k$ and set $Z$ exist, the algorithm allocates $Z$ to $i$ (without breaking EFX), while prioritizing vertices with a higher index to preserve Property~\eqref{propSec4:Non-EnviedNeighbors}. This Algorithm satisfies  Property~\eqref{propSec3:No-envyToSingleEdges}, as long as Property \eqref{propSec3:EFXorientation} was initially satisfied. Algorithm~\ref{Algo:OrientationForMultiplicityGreaterThan2} initially calls \FixPTM\ to  guarantee Properties \eqref{propSec4:EFXorientation}-\eqref{propSec4:No-envytoPatches}, and then each vertex, except vertex $0$, is offered its relevant unallocated edges resulting in satisfying Property~\eqref{propSec4:No-envyToFlowers}; the Algorithm~\FixPTM~is executed any time an offer is accepted so that Property~\eqref{propSec4:No-envytoPatches} is preserved. Algorithm~\ref{Algo:OrientationForMultiplicityGreaterThan2} terminates with all four properties satisfied.

The final allocation (Algorithm~\ref{Algo:FinalAllocationMultipleEdges}) differs from the one in Section~\ref{section:EFXhyper} because now there may be unallocated edges with non-envied endpoints that cannot be oriented. In Section~\ref{section:EFXhyper}, the unallocated edges that could not be oriented were given to vertex $0$ without causing envy towards $0$, since all vertices that had positive value for those edges either were not neighbors to vertex $0$ or they had received the \emph{single} edge shared with vertex $0$ (Property~\eqref{propSec3:Non-EnviedNeighbors}); in both cases this meant that they had zero value for $X_0$. In the multi-hypergraph setting, it may be that for vertices that value positively $X_0$ there are unallocated edges incident to them that cannot be oriented. In that case, no property guarantees that allocating those edges to vertex $0$ would not create envy towards $0$. For those vertices, we find alternative vertices to ``park" their relevant unallocated edges. Next we describe how we do that.

If $X_0 = \emptyset$ at this phase, we allocate all unallocated edges to vertex 0, and due to Property~\eqref{propSec4:No-envyToFlowers} no envy will be created. If $X_0 \neq \emptyset$, %
then $X_0\subseteq \Patch_e$, for some $e\in E$, as $i=0$  is not considered in line 3 of Algorithm \ref{Algo:OrientationForMultiplicityGreaterThan2}. Since $|\Patch_e|$ is no more than the size of $e$ minus 2 (by Theorem~\ref{thm:girth4_multi}'s assumption), there would be two vertices $j_1,j_2 \in e$ that are not allocated any edge from $\Patch_e$. By Property \eqref{propSec4:Non-EnviedNeighbors}, those two vertices are non-envied, and by Observations~\ref{obs:at-most-1-relevant-common-edge}~and~\ref{obs:no_Triangles}, each of the $j_1,j_2$ is allocated edges that are irrelevant to any vertex of $e$ and their neighbors. So, allocating $U_{i}(\X)$ to $j_1$, for any $i\in e$ different from $j_1$ and $0$, and $U_{j_1}(\X)$ to $j_2$, creates no further envy due to Properties~\eqref{propSec4:No-envytoPatches}-\eqref{propSec4:No-envyToFlowers}. For any other vertex $i\notin e$, $X_0$ is irrelevant for $i$, and so allocating $U_i(\X)$ to vertex $0$ again creates no further envy due to Property~\eqref{propSec4:No-envyToFlowers}.

\subsection{Construction of the EFX allocation.}

The complete EFX allocation is derived by the sequential execution of Algorithms~\ref{Algo:OrientationForMultiplicityGreaterThan2} and \ref{Algo:FinalAllocationMultipleEdges}. Algorithm~\ref{Algo:OrientationForMultiplicityGreaterThan2} calls \FixPTM\ (Algorithm~\ref{Algo:AllocatingMinimalEnviedSubset}) to maintain Property~\eqref{propSec4:No-envytoPatches}. To succeed this, \FixPTM\ makes use of Subroutine~\ref{Algo:FindingMinimalEnviedSubset} that returns a minimal set $Z$ of same patch edges that a vertex prefers to its bundle, meaning that no other vertex envies any proper subset of $Z$. Subroutine~\ref{Algo:FindingMinimalEnviedSubset} is the same with Algorithm 3 from \cite{CKMS21}; we give Subroutine~\ref{Algo:FindingMinimalEnviedSubset} in Appendix~\ref{app:subroutine7} for completeness.

\begin{algorithm}[t]
\caption{\FixPTM}
\label{Algo:AllocatingMinimalEnviedSubset}
\raggedright\textbf{Input:} An allocation $\X$ satisfying Property (1). \\
\textbf{Output:} An allocation $\X$ satisfying Properties \eqref{propSec4:EFXorientation} and \eqref{propSec4:No-envytoPatches}.\\
\begin{algorithmic}[1]
\WHILE{ $\exists$ $i\in V$, $e\in E$: $v_i(X_i) < v_i(U_e(\X))$}
    \STATE $Z=$ Subroutine~\ref{Algo:FindingMinimalEnviedSubset} ($L=e$, $S=U_e(\X)$).
    \STATE $k = \arg\max_j\{v_j(Z) > v_j(X_j)\}$
    \STATE $X_k \gets Z$
\ENDWHILE
\RETURN $\X$
\end{algorithmic}
\end{algorithm}

\begin{restatable}{lemma}{lemFixPTM}
    \label{lemma:FTXM}
    When \FixPTM\ is called for an EFX orientation outputs an EFX orientation satisfying Property \eqref{propSec4:No-envytoPatches}.
\end{restatable}
\begin{proof}
First note that \FixPTM\ will terminate since the bundles (and thus the values) that a vertex may get are finite, and at every execution of the while, the value of some vertex for the updated allocation strictly increases. 

    To see that Property \eqref{propSec4:No-envytoPatches} is satisfied in the resulting allocation, observe that the condition of while is the negation of Property \eqref{propSec4:No-envytoPatches}, so \FixPTM\ terminates by satisfying it.

The allocating step of line 4  orients a set of unallocated edges to one of its adjacent vertices; this is because the set $Z$ returned by Subroutine~\ref{Algo:FindingMinimalEnviedSubset} (line 2) is a subset of $U_e(\X)$, and $k$ defined in line 3 belongs to $e$ since $v_k(Z)>v_k(X_k)\geq 0$. This keeps the allocation an orientation. This allocation further remains an EFX allocation after each execution of the while, since no vertex envied a proper subset of $Z$ and also the value of the vertices may only increase during the execution of \FixPTM, and therefore no further envy may appear in the future.
\end{proof}

Algorithm~\ref{Algo:OrientationForMultiplicityGreaterThan2} extends Property \eqref{propSec4:No-envytoPatches} to Property \eqref{propSec4:No-envyToFlowers} for all vertices but $0$, while preserving Properties \eqref{propSec4:EFXorientation} and \eqref{propSec4:Non-EnviedNeighbors}.

\begin{algorithm}[t]
\caption{Orienting Edges}

\label{Algo:OrientationForMultiplicityGreaterThan2}
\raggedright\textbf{Input:} A multi-hypergraph $G$ of girth at least 4, where edges related to vertex $0$ have multiplicity at most that edge's size minus 2.\\
\textbf{Output:} An allocation $\X$ satisfying Properties \eqref{propSec4:EFXorientation}-\eqref{propSec4:No-envyToFlowers}.\\
\begin{algorithmic}[1]
    \STATE \FixPTMX
    \WHILE{ $\exists$ vertex $i \neq 0$: $v_i(U_i(\X)) > v_i(X_i)$}
        \STATE $X_i \gets U_i(\X)$
        \STATE \FixPTMX
    \ENDWHILE

\end{algorithmic}
\end{algorithm}

\begin{restatable}{lemma}{lempropAllMulti}
    \label{lem:prop1-4Multi}
    After the termination of Algorithm \ref{Algo:OrientationForMultiplicityGreaterThan2}, Properties \eqref{propSec4:EFXorientation}-\eqref{propSec4:No-envyToFlowers} are satisfied.
\end{restatable}

\begin{proof}

Lemma~\ref{lem:prop1-4Multi} is proven by combining the following three lemmas. Lemma \ref{lem:prop124MultiAppendix} shows that Properties (1)~and~\eqref{propSec4:No-envyToFlowers} hold for the allocation that Algorithm~\ref{Algo:OrientationForMultiplicityGreaterThan2} returns, then Lemmas~\ref{lemma:0non-enviedMultiAppendix}~and~\ref{lem:Prop3_2ndPartMultiAppendix} show that Property~\eqref{propSec4:Non-EnviedNeighbors} hold for vertex 0, and any neighbor of vertex 0, respectively. Moreover, since Algorithm~\ref{Algo:OrientationForMultiplicityGreaterThan2} terminates with a call of \FixPTM, by Lemma~\ref{lemma:FTXM}, Property~\eqref{propSec4:No-envytoPatches} is also satisfied.
\end{proof}

\begin{lemma}
    \label{lem:prop124MultiAppendix}
    After the termination of Algorithm \ref{Algo:OrientationForMultiplicityGreaterThan2}, Properties \eqref{propSec4:EFXorientation} and \eqref{propSec4:No-envyToFlowers} are satisfied.
\end{lemma}

\begin{proof}
Algorithm \ref{Algo:OrientationForMultiplicityGreaterThan2} will terminate since both \FixPTM{} and the allocating step of line 3 only strictly increase the value that a vertex gets and the bundles (and thus the values) that a vertex may get are finite. The condition of the while loop of Algorithm \ref{Algo:OrientationForMultiplicityGreaterThan2} is the negation of Property~\eqref{propSec4:No-envyToFlowers}. Since the condition of the while loop must be false in order for the algorithm to terminate, it means that Property~\eqref{propSec4:No-envyToFlowers} will hold.

It remains to show that Property \eqref{propSec4:EFXorientation} holds when Algorithm \ref{Algo:OrientationForMultiplicityGreaterThan2} terminates. First thing to note is that \FixPTM\ outputs an orientation if its input is an orientation and the allocating step of line 3 orients edges (towards one of its adjacent vertices), and thus the resulting allocation is an orientation (since the initial allocation is the empty allocation which is trivially an orientation). To see that the allocation remains EFX it suffices to show that the allocating step of line 3 does not break the EFX property, since by Lemma~\ref{lemma:FTXM} \FixPTM\ does not break the EFX property. Consider a vertex $i$ that gets its bundle changed by the allocating step of line 3. 
For any other vertex $j$, either $j$ is not $i$'s neighbor, so will not envy $i$ for receiving $U_i(\X)$, or by Observation \ref{obs:at-most-1-relevant-common-edge}, $U_i(\X)$ contains exactly one patch relevant to $j$, say $\Patch_e$. Since \FixPTM\ terminated right before the update of line 3, vertex $j$ does not prefer the unallocated set of that patch, i.e., $v_j(X_j)\geq v_j(U_e(\X))=v_j(U_i(\X))$, and thus $i$ will be non-envied. Also the value of the vertices may only increase during the execution of Algorithm~\ref{Algo:OrientationForMultiplicityGreaterThan2}, and therefore no further envy may appear and the EFX property will not break.
\end{proof}

\begin{lemma}
    \label{lemma:0non-enviedMultiAppendix}
    After the termination of  Algorithm \ref{Algo:OrientationForMultiplicityGreaterThan2} vertex $0$ is non-envied.
\end{lemma}

\begin{proof}
Initially it is $X_0=\emptyset$ and thus vertex $0$ is non-envied. Note that vertex $0$ may only change its allocated bundle during the procedure \FixPTM. Any time that \FixPTM\ is called, in order for vertex $0$ to be considered for a change in its bundle it must be that all other vertices do not prefer the set of the edges that vertex $0$ is allocated, since all of them have higher index than $0$. This directly implies that if vertex $0$ is allocated a set of edges during \FixPTM\ no vertex will envy vertex $0$ and thus, since the values of the vertices only increase, vertex $0$ remains non-envied.
\end{proof}

\begin{lemma}\label{lem:Prop3_2ndPartMultiAppendix}
    After the termination of Algorithm \ref{Algo:OrientationForMultiplicityGreaterThan2}, any neighbor $i$ of vertex $0$ may only be envied if it is allocated some edge(s) that it shares with $0$.
\end{lemma}
\begin{proof}

Consider any vertex $i \in [n_0]$, i.e., a neighbor of vertex $0$, and the bundle $X_i$ allocated to $i$ after the termination of Algorithm \ref{Algo:OrientationForMultiplicityGreaterThan2}.
If $X_i = \emptyset$, then $i$ is trivially non-envied, otherwise we distinguish between the cases $X_i$ is a subset of a single patch or not. Consider first the case that $X_i$ contains edges from at least two different patches, then $i$ is non-envied due to Property~\eqref{propSec4:EFXorientation}. To see this suppose on the contrary that $i$ is envied by some vertex $j$. Since $i$ receives only relevant edges to itself, $j$ should be $i$'s neighbor. By Observation~\ref{obs:at-most-1-relevant-common-edge}, $i$ and $j$ share only one patch; let it be $\Patch_e$. Since $j$ envies $i$, it holds that $v_j(X_i\cap \Patch_e)>v_j(X_j)$, and since $X_i$ contains edges from at least two different patches, it holds that $X_i\setminus \Patch_e$ is not empty. So, after the removal of any edge from $X_i\setminus \Patch_e$, $j$ would still envy $i$, which violates the EFX condition and therefore Property~\eqref{propSec4:EFXorientation}. Hence, in this case, $i$ is not envied.

Suppose now that $X_i$ contains edges from only one patch, so $X_i\subseteq \Patch_e$, for some edge $e$, and suppose that $\Patch_e$ is not the patch that $i$ shares with vertex $0$. Since Algorithm~\ref{Algo:OrientationForMultiplicityGreaterThan2} only orients edges, $X_i$ is relevant to $i$. We next show that $\Patch_e$ (and so $X_i$) cannot be relevant to any other neighbor of vertex $0$, and therefore none of them may envy $i$. Let $j \in [n_0]$ be some neighbor of $0$ belonging to the patch that $0$ and $i$ share. By Observation~\ref{obs:at-most-1-relevant-common-edge}, $\Patch_e$ is irrelevant to $j$. Let $j \in [n_0]$ now be some neighbor of 0 not belonging to the patch $i$ and $0$ share. Due to Observation~\ref{obs:no_Triangles}, vertices $i,j$ cannot both belong in the same edge, let's say $e'$, or else they would have vertex $0$ as a common neighbor outside $e'$. Thus, $\Patch_e$ is irrelevant to $j$ in that case as well. Vertices indexed higher than $n_0$ did not envy $i$ at the time it received $X_i$ during the execution of Algorithm~\ref{Algo:OrientationForMultiplicityGreaterThan2}: If $X_i$ was allocated to $i$ in a run of \FixPTM, vertex $i$ was considered as the highest indexed vertex that preferred $X_i$ to what it got, implying that all vertices with index higher than $i$ preferred their bundles over every set of unallocated edges in some patch, which includes $X_i$. If $X_i$ was allocated to $i$ in the while of Algorithm~\ref{Algo:OrientationForMultiplicityGreaterThan2}, then Property~\eqref{propSec4:No-envytoPatches} (from the previous \FixPTM\ run) guarantees that nobody envied $i$ at that point. Note that $i$ remains non-envied until the termination of the Algorithm~\ref{Algo:OrientationForMultiplicityGreaterThan2}, since the vertices' value may only increase. 
\end{proof}

Algorithm~\ref{Algo:FinalAllocationMultipleEdges} allocates any unallocated edges from the allocation of Algorithm~\ref{Algo:OrientationForMultiplicityGreaterThan2}. If $X_0 = \emptyset$, then all remaining unallocated edges are allocated to $X_0$. If $X_0 \neq \emptyset$, we carefully select two non-envied neighbors of vertex $0$ and the unallocated edges are allocated among those two vertices and $0$ in such a way that no further envy is created.

\begin{algorithm}[t]
\caption{Complete allocation}
\label{Algo:FinalAllocationMultipleEdges}
\raggedright\textbf{Input:} The allocation $\X$ returned by Algorithm~\ref{Algo:OrientationForMultiplicityGreaterThan2}.\\
\textbf{Output:} A complete EFX allocation $\X$.\\
\begin{algorithmic}[1]

\IF{ $X_0 \neq \emptyset$}
    \STATE Let $e$ be the edge such that $X_0\subseteq \Patch_e$.
    \STATE Let $j_1,j_2 \in e$ be two non-envied vertices where $X_j \cap \Patch_e=\emptyset$, for $j\in\{j_1,j_2\}$.
    \FOR{ every $i\in e\setminus\{j_1,0\}$ }
        \STATE $X_{j_1} \gets X_{j_1} \cup U_i(\X)$
    \ENDFOR
    \STATE $X_{j_2} \gets X_{j_2} \cup U_{j_1}(\X)$
\ENDIF
\STATE $X_0 \gets X_0 \cup U(\X)$
\end{algorithmic}
\end{algorithm}

\begin{restatable}{lemma}{lemmaFinalAllocationMulti}
\label{lemma:FinalAllocationMulti}
Let $\X$ be  the allocation returned by Algorithm~\ref{Algo:OrientationForMultiplicityGreaterThan2}. Then, if 
    Algorithm~\ref{Algo:FinalAllocationMultipleEdges} takes $\X$ as input, it returns a complete EFX allocation.
    
\end{restatable}
\begin{proof}
    
    If $X_0 = \emptyset$ at this phase, then we allocate all unallocated edges to vertex 0, and due to Property~\eqref{propSec4:No-envyToFlowers} no further envy will be created. 
    
    For the case that $X_0 \neq \emptyset$, we first establish that vertices $j_1,j_2$ (line 3) exist. Note that vertex $0$ may be only allocated edges from a single patch, so suppose $X_0\subseteq \Patch_e$, for some edge $e\in E$. Since $|\Patch_e|$ is no more than the size of $e$ minus 2 (by Theorem~\ref{thm:girth4_multi}'s assumption), there would be two vertices $j_1,j_2 \in e$, different from vertex $0$, that are not allocated any edge from $\Patch_e$. By Property \eqref{propSec4:Non-EnviedNeighbors} those two vertices are non-envied.
    
    We next show that $j_1$ will not be envied after the execution of lines 4-6. We show this by arguing that any vertex that finds relevant any edge allocated to $j_1$ during this procedure, found $X_{j_1}$ (as it was prior the execution of lines 4-6) irrelevant. Consider some vertex $k$ that finds relevant some of the edge(s) of $U_i(\X)$, for some $i\in e\setminus \{j_1,0\}$. If $k\in e$ then by Observation~\ref{obs:at-most-1-relevant-common-edge}, $k$ and $j_1$ do not share any other edge but $\Patch_e$, therefore $v_k(X_{j_1})=0$. Otherwise, $k \notin e$ is some neighbor of vertex $i$. By Observation~\ref{obs:no_Triangles}, $i$ and $j_1$ do not have any common neighbor outside $e$, therefore, $k$ is not $j_1$'s neighbor and so, $k$ finds $X_{j_1}$ irrelevant, i.e.,  $v_k(X_{j_1})=0$, in this case as well. Therefore, allocating $\cup_{i \in e\setminus \{j_1,0\}}U_i(\X)$ to $j_1$ do not make any vertex $k\neq 0$ envious towards $j_1$ due to Property~\eqref{propSec4:No-envyToFlowers}. Regarding vertex $0$, note that by Observation~\ref{obs:at-most-1-relevant-common-edge} vertex $0$ do not share any other edge apart from $\Patch_e$ with any other vertex $i \in e$, so the relevant edges for vertex $0$ in $\cup_{i \in e\setminus \{j_1,0\}}U_i(\X)$ are the $U_e(\X)$; due to Property~\ref{propSec4:No-envytoPatches}, vertex $0$ will also not envy $j_1$ after the execution of lines 4-6.    

    The same arguments hold for the allocation at line 7; no envy would be created towards vertex $j_2$. 
Regarding the allocation at line 9,  for any other vertex $i\notin e$, $X_0$ is irrelevant for $i$, since vertex $0$ is allocated edges only from patch $\Patch_e$, and so allocating $U_i(\X)$ to vertex $0$ creates no further envy due to Property~\eqref{propSec4:No-envyToFlowers}.
\end{proof}

The running time of Subroutine~\ref{Algo:FindingMinimalEnviedSubset} is polynomial on the number of different values of the social welfare, yet, this number %
can be exponential to the size of the input.

\begin{restatable}{lemma}{lemEFXAllocMultiInpseudopolytime}
    The construction of the EFX allocation by running Algorithms~\ref{Algo:OrientationForMultiplicityGreaterThan2}~and~\ref{Algo:FinalAllocationMultipleEdges} needs pseudo-polynomial time.
\end{restatable}

\begin{proof}
    We first argue that each iteration of the while of the Procedure~\FixPTM\ needs polynomial time. Checking the condition of the while needs time complexity $O(n^2)$, as it needs to check for each of the $n$ vertices, at most $n-1$ different patches (there at most $n-1$ relevant patches for each vertex which coincides with its degree in the underlying simple hypergraph). Inside the while of \FixPTM, the Subroutine~\ref{Algo:FindingMinimalEnviedSubset} is executed for a subset of vertices and edges and according to \cite{CKMS21} it needs $O(mn)$ time. Therefore, each execution of the while needs time $O((m+n)n)$. 

    We next argue that each iteration of the while of Algorithm~\ref{Algo:OrientationForMultiplicityGreaterThan2}, excluding the call of \FixPTM, needs polynomial time. Indeed, checking the condition of the while needs time complexity $O(n)$, as it performs one comparison per vertex. 

    Note that any time the while of \FixPTM\ is executed the sum of the vertices' value (social welfare) \emph{strictly} increases. The same holds in line 3 of the while of Algorithm~\ref{Algo:OrientationForMultiplicityGreaterThan2}. Therefore, after polynomially many steps an increment on the social welfare happens. Following the analysis of \cite{CKMS21}, that also used the social welfare as the potential of their algorithm, Algorithm~\ref{Algo:OrientationForMultiplicityGreaterThan2} needs pseudo-polynomial time.

    Since Algorithm~\ref{Algo:FinalAllocationMultipleEdges} allocates at most $m$ edges, while it requires time $O(n)$ to find $j_1,j_2$, the time complexity is $O(m+n)$.
    
    Hence, overall, executing Algorithms~\ref{Algo:OrientationForMultiplicityGreaterThan2}~and~\ref{Algo:FinalAllocationMultipleEdges} needs pseudo-polynomial time.
\end{proof}

\section{Conclusion}

Our results build upon the existing literature on graph setting based on the definition of \cite{EFXsimplegraphs}, and push the state-of-the-art even further towards the general problem of EFX existence, which is equivalent to the multi-hypergraph setting without any restrictions. We introduce a different angle on the approach of the problem by designating from the start special vertices to allocate non-oriented edges. 
We remark that our approach is quite general and is applied to arbitrarily heterogeneous monotone valuations, and any restrictions come merely from the structure of the multi-hypergraph. %
Future directions include lifting the restrictions in the structure of the multi-hypergraph, or even improve complexity for restricted cases.

\section{Acknowledgements}
The research project is implemented in the framework of H.F.R.I call “Basic research Financing (Horizontal support
of all Sciences)” under the National Recovery and Resilience Plan “Greece 2.0” funded by the European Union-
NextGenerationEU (H.F.R.I. Project Number:15635). This work has been partially supported by project MIS 5154714 of the National Recovery and Resilience
Plan Greece 2.0 funded by the European Union under the NextGenerationEU Program.

\bibliographystyle{ACM-Reference-Format}
\bibliography{refs}

@inproceedings{afshinmehrMultigraphs2026efx,
  title     = {EFX Allocations Exist on Multi-Graphs},
  author    = {Mahyar Afshinmehr and Arash Ashuri and Pouria Mahmoudkhan and Kurt Mehlhorn and Amir Mohammad Shahrezaei},
  booktitle = {Proceedings of the 27th ACM Conference on Economics and Computation},
  year      = {2026},
  location  = {Rome, Italy},

}

@inproceedings{EFXmultigraphsChristodoulou,
  title     = {EFX allocations on multigraphs},
  author    = {Giorgos Christodoulou and Symeon Mastrakoulis and Alkmini Sgouritsa and Minas Marios Sotiriou},
  booktitle = {Proceedings of the 27th ACM Conference on Economics and Computation},
  year      = {2026},
  location  = {Rome, Italy},
  year		= {2026}
}

@inproceedings{GMT14,
  author       = {Laurent Gourv{\`{e}}s and
                  J{\'{e}}r{\^{o}}me Monnot and
                  Lydia Tlilane},
  editor       = {Torsten Schaub and
                  Gerhard Friedrich and
                  Barry O'Sullivan},
  title        = {Near Fairness in Matroids},
  booktitle    = {{ECAI} 2014 - 21st European Conference on Artificial Intelligence,
                  18-22 August 2014, Prague, Czech Republic - Including Prestigious
                  Applications of Intelligent Systems {(PAIS} 2014)},
  series       = {Frontiers in Artificial Intelligence and Applications},
  volume       = {263},
  pages        = {393--398},
  publisher    = {{IOS} Press},
  year         = {2014},
  url          = {https://doi.org/10.3233/978-1-61499-419-0-393},
  doi          = {10.3233/978-1-61499-419-0-393},
  bibsource    = {dblp computer science bibliography, https://dblp.org}
}

@inproceedings{TractableGraphStructures,
author = {Bla\v{z}ej, V\'{a}clav and Gupta, Sushmita and Ramanujan, M. S. and Strulo, Peter},
title = {Tractable Graph Structures in\&nbsp;EFX Orientation},
year = {2025},
isbn = {978-3-032-03638-4},
publisher = {Springer-Verlag},
address = {Berlin, Heidelberg},
url = {https://doi.org/10.1007/978-3-032-03639-1_10},
doi = {10.1007/978-3-032-03639-1_10},
booktitle = {Algorithmic Game Theory: 18th International Symposium, SAGT 2025, Bath, UK, September 2–5, 2025, Proceedings},
pages = {175–190},
numpages = {16},
location = {Bath, United Kingdom}
}

@misc{kanellopoulos2025efxorientationsparameterizedcomplexity,
      title={EF(X) Orientations: A Parameterized Complexity Perspective}, 
      author={Sotiris Kanellopoulos and Edouard Nemery and Christos Pergaminelis and Minas Marios Sotiriou and Manolis Vasilakis},
      year={2025},
      eprint={2512.25033},
      archivePrefix={arXiv},
      primaryClass={cs.DS},
      url={https://arxiv.org/abs/2512.25033}, 
}

@inproceedings{afshinmehr2024efxallocationsorientationsbipartite,
  author       = {Mahyar Afshinmehr and
                  Alireza Danaei and
                  Mehrafarin Kazemi and
                  Kurt Mehlhorn and
                  Nidhi Rathi},
  title        = {{EFX} Allocations and Orientations on Bipartite Multi-graphs: {A}
                  Complete Picture},
  booktitle    = {{AAMAS}},
  pages        = {32--40},
  publisher    = {International Foundation for Autonomous Agents and Multiagent Systems
                  / {ACM}},
  year         = {2025}
}

@Misc{		  afshinmehr2025efxallocationsexisttrianglefree,
  title		= {EFX Allocations Exist on Triangle-Free Multi-Graphs},
  author	= {Mahyar Afshinmehr and Arash Ashuri and Pouria Mahmoudkhan
		  and Kurt Mehlhorn},
  year		= {2025},
  eprint	= {2512.21644},
  archiveprefix	= {arXiv},
  primaryclass	= {cs.GT},
  url		= {https://arxiv.org/abs/2512.21644}
}

@InProceedings{ akrami2022efxallocationssimplificationsimprovements,
  author	= {Hannaneh Akrami and Noga Alon and Bhaskar Ray Chaudhury
		  and Jugal Garg and Kurt Mehlhorn and Ruta Mehta},
  editor	= {Kevin Leyton{-}Brown and Jason D. Hartline and Larry
		  Samuelson},
  title		= {{EFX:} {A} Simpler Approach and an (Almost) Optimal
		  Guarantee via Rainbow Cycle Number},
  booktitle	= {Proceedings of the 24th {ACM} Conference on Economics and
		  Computation, {EC} 2023, London, United Kingdom, July 9-12,
		  2023},
  pages		= {61},
  publisher	= {{ACM}},
  year		= {2023},
  url		= {https://doi.org/10.1145/3580507.3597799},
  doi		= {10.1145/3580507.3597799},
  bibsource	= {dblp computer science bibliography, https://dblp.org}
}

@InProceedings{	  almostfullefxforfouragents,
  author	= {Ben Berger and Avi Cohen and Michal Feldman and Amos
		  Fiat},
  title		= {Almost Full {EFX} Exists for Four Agents},
  booktitle	= {Thirty-Sixth {AAAI} Conference on Artificial Intelligence,
		  {AAAI} 2022, Thirty-Fourth Conference on Innovative
		  Applications of Artificial Intelligence, {IAAI} 2022, The
		  Twelveth Symposium on Educational Advances in Artificial
		  Intelligence, {EAAI} 2022 Virtual Event, February 22 -
		  March 1, 2022},
  pages		= {4826--4833},
  publisher	= {{AAAI} Press},
  year		= {2022},
  url		= {https://doi.org/10.1609/aaai.v36i5.20410},
  doi		= {10.1609/AAAI.V36I5.20410},
  bibsource	= {dblp computer science bibliography, https://dblp.org}
}

@InProceedings{	  amanatidis2024pushingfrontierapproximateefx,
  author	= {Georgios Amanatidis and Aris Filos-Ratsikas and Alkmini
		  Sgouritsa},
  title		= {Pushing the Frontier on Approximate EFX Allocations},
  booktitle	= {Proceedings of the 25th {ACM} Conference on Economics and
		  Computation, {EC}},
    year		= {2024}
}

@InProceedings{	  aziz2016,
  author	= {Haris Aziz and Simon Mackenzie},
  editor	= {Daniel Wichs and Yishay Mansour},
  title		= {A discrete and bounded envy-free cake cutting protocol for
		  four agents},
  booktitle	= {Proceedings of the 48th Annual {ACM} {SIGACT} Symposium on
		  Theory of Computing, {STOC} 2016},
  publisher	= {{ACM}},
  year		= {2016},
  url		= {https://doi.org/10.1145/2897518.2897522},

}

@InProceedings{	  babaioff2020fairtruthfulmechanismsdichotomous,
  author	= {Moshe Babaioff and Tomer Ezra and Uriel Feige},
  title		= {Fair and Truthful Mechanisms for Dichotomous Valuations},
  booktitle	= {Thirty-Fifth {AAAI} Conference on Artificial Intelligence,
		  {AAAI} 2021, Thirty-Third Conference on Innovative
		  Applications of Artificial Intelligence, {IAAI} 2021, The
		  Eleventh Symposium on Educational Advances in Artificial
		  Intelligence, {EAAI} 2021, Virtual Event, February 2-9,
		  2021} ,
  pages		= {5119--5126},
  publisher	= {{AAAI} Press},
  year		= {2021},
  url		= {https://doi.org/10.1609/aaai.v35i6.16647},
  doi		= {10.1609/AAAI.V35I6.16647},
  bibsource	= {dblp computer science bibliography, https://dblp.org}
}

@InProceedings{	  berendsohn2022fixedpointcyclesefxallocations,
  author	= {Benjamin Aram Berendsohn and Simona Boyadzhiyska and
		  L{\'{a}}szl{\'{o}} Kozma},
  editor	= {Stefan Szeider and Robert Ganian and Alexandra Silva},
  title		= {Fixed-Point Cycles and Approximate {EFX} Allocations},
  booktitle	= {47th International Symposium on Mathematical Foundations
		  of Computer Science, {MFCS} 2022, August 22-26, 2022,
		  Vienna, Austria},
  series	= {LIPIcs},
  volume	= {241},
  pages		= {17:1--17:13},
  publisher	= {Schloss Dagstuhl - Leibniz-Zentrum f{\"{u}}r Informatik},
  year		= {2022},
  url		= {https://doi.org/10.4230/LIPIcs.MFCS.2022.17},
  doi		= {10.4230/LIPICS.MFCS.2022.17},
  bibsource	= {dblp computer science bibliography, https://dblp.org}
}

@Book{		  berge1973graphs,
  author	= {Claude Berge},
  title		= {Graphs and Hypergraphs},
  year		= {1973},
  publisher	= {North-Holland},
  address	= {Amsterdam}
}

@InProceedings{bhaskar2024efxallocationsmultigraphclasses,
  author =	{Bhaskar, Umang and Pandit, Yeshwant},
  title =	{{Extending EFX Allocations to Further Multi-Graph Classes}},
  booktitle =	{45th IARCS Annual Conference on Foundations of Software Technology and Theoretical Computer Science (FSTTCS 2025)},
  pages =	{15:1--15:18},
  series =	{Leibniz International Proceedings in Informatics (LIPIcs)},
  ISBN =	{978-3-95977-406-2},
  ISSN =	{1868-8969},
  year =	{2025},
  volume =	{360},
  editor =	{Aiswarya, C. and Mehta, Ruta and Roy, Subhajit},
  publisher =	{Schloss Dagstuhl -- Leibniz-Zentrum f{\"u}r Informatik},
  address =	{Dagstuhl, Germany},
  URL =		{https://drops.dagstuhl.de/entities/document/10.4230/LIPIcs.FSTTCS.2025.15},
  URN =		{urn:nbn:de:0030-drops-250958},
  doi =		{10.4230/LIPIcs.FSTTCS.2025.15}
}

@InProceedings{	  budish,
  author	= {Eric Budish},
  editor	= {Moshe Dror and Greys Sosic},
  title		= {The combinatorial assignment problem: approximate
		  competitive equilibrium from equal incomes},
  booktitle	= {Proceedings of the Behavioral and Quantitative Game Theory
		  - Conference on Future Directions, {BQGT} '10, Newport
		  Beach, California, USA, May 14-16, 2010},
  pages		= {74:1},
  publisher	= {{ACM}},
  year		= {2010},
  url		= {https://doi.org/10.1145/1807406.1807480},
  doi		= {10.1145/1807406.1807480},
  bibsource	= {dblp computer science bibliography, https://dblp.org}
}

@InProceedings{	  cgh19,
  author	= {Ioannis Caragiannis and Nick Gravin and Xin Huang},
  editor	= {Anna R. Karlin and Nicole Immorlica and Ramesh Johari},
  title		= {Envy-Freeness Up to Any Item with High Nash Welfare: The
		  Virtue of Donating Items},
  booktitle	= {Proceedings of the 2019 {ACM} Conference on Economics and
		  Computation, {EC} 2019, Phoenix, AZ, USA, June 24-28,
		  2019},
  pages		= {527--545},
  publisher	= {{ACM}},
  year		= {2019},
  url		= {https://doi.org/10.1145/3328526.3329574},
  doi		= {10.1145/3328526.3329574},
  bibsource	= {dblp computer science bibliography, https://dblp.org}
}

@Article{	  cgm24,
  author	= {Bhaskar Ray Chaudhury and Jugal Garg and Kurt Mehlhorn},
  title		= {{EFX} Exists for Three Agents},
  journal	= {J. {ACM}},
  volume	= {71},
  number	= {1},
  pages		= {4:1--4:27},
  year		= {2024},
  url		= {https://doi.org/10.1145/3616009},
  doi		= {10.1145/3616009},
  bibsource	= {dblp computer science bibliography, https://dblp.org}
}

@InProceedings{	  cgmmm21,
  author	= {Bhaskar Ray Chaudhury and Jugal Garg and Kurt Mehlhorn and
		  Ruta Mehta and Pranabendu Misra},
  editor	= {P{\'{e}}ter Bir{\'{o}} and Shuchi Chawla and Federico
		  Echenique},
  title		= {Improving {EFX} Guarantees through Rainbow Cycle Number},
  booktitle	= {{EC} '21: The 22nd {ACM} Conference on Economics and
		  Computation, Budapest, Hungary, July 18-23, 2021},
  pages		= {310--311},
  publisher	= {{ACM}},
  year		= {2021},
  url		= {https://doi.org/10.1145/3465456.3467605},
  doi		= {10.1145/3465456.3467605},
  bibsource	= {dblp computer science bibliography, https://dblp.org}
}

@Article{	  ckms21,
  author	= {Bhaskar Ray Chaudhury and Telikepalli Kavitha and Kurt
		  Mehlhorn and Alkmini Sgouritsa},
  title		= {A Little Charity Guarantees Almost Envy-Freeness},
  journal	= {{SIAM} J. Comput.},
  volume	= {50},
  number	= {4},
  pages		= {1336--1358},
  year		= {2021},
  url		= {https://doi.org/10.1137/20M1359134},
  doi		= {10.1137/20M1359134},
  bibsource	= {dblp computer science bibliography, https://dblp.org}
}

@InProceedings{	  dblp:conf/aaai/hosseinisvx21,
  author	= {Hadi Hosseini and Sujoy Sikdar and Rohit Vaish and Lirong
		  Xia},
  title		= {Fair and Efficient Allocations under Lexicographic
		  Preferences},
  booktitle	= {Thirty-Fifth {AAAI} Conference on Artificial Intelligence,
		  {AAAI} 2021, Thirty-Third Conference on Innovative
		  Applications of Artificial Intelligence, {IAAI} 2021, The
		  Eleventh Symposium on Educational Advances in Artificial
		  Intelligence, {EAAI} 2021, Virtual Event, February 2-9,
		  2021},
  pages		= {5472--5480},
  publisher	= {{AAAI} Press},
  year		= {2021},
  url		= {https://doi.org/10.1609/aaai.v35i6.16689},
  doi		= {10.1609/AAAI.V35I6.16689},
  bibsource	= {dblp computer science bibliography, https://dblp.org}
}

@inproceedings{deligkas2024ef1efxorientations,
  title     = {EF1 and EFX Orientations},
  author    = {Deligkas, Argyrios and Eiben, Eduard and Goldsmith, Tiger-Lily and Korchemna, Viktoriia},
  booktitle = {Proceedings of the Thirty-Fourth International Joint Conference on
               Artificial Intelligence, {IJCAI-25}},
  publisher = {International Joint Conferences on Artificial Intelligence Organization},
  pages     = {8},
  year      = {2025},
  note      = {Main Track},
  doi       = {10.24963/ijcai.2025/7},
  url       = {https://doi.org/10.24963/ijcai.2025/7},
}

@Article{	  efxcara,
  author	= {Ioannis Caragiannis and David Kurokawa and Herv{\'{e}}
		  Moulin and Ariel D. Procaccia and Nisarg Shah and Junxing
		  Wang},
  title		= {The Unreasonable Fairness of Maximum Nash Welfare},
  journal	= {{ACM} Trans. Economics and Comput.},
  volume	= {7},
  number	= {3},
  pages		= {12:1--12:32},
  year		= {2019},
  url		= {https://doi.org/10.1145/3355902},
  doi		= {10.1145/3355902},
  bibsource	= {dblp computer science bibliography, https://dblp.org}
}

@InProceedings{	  efxexistsfor3typesofadditiveagents,
  author	= {Hv, Vishwa Prakash and Ghosal, Pratik and Nimbhorkar,
		  Prajakta and Varma, Nithin},
  title		= {EFX Exists for Three Types of Agents},
  year		= {2025},
  isbn		= {9798400719431},
  publisher	= {Association for Computing Machinery},
  address	= {New York, NY, USA},
  url		= {https://doi.org/10.1145/3736252.3742509},
  doi		= {10.1145/3736252.3742509},
  booktitle	= {Proceedings of the 26th ACM Conference on Economics and
		  Computation},
  pages		= {101–128},
  numpages	= {28},
  location	= {Stanford University, Stanford, CA, USA},
  series	= {EC '25}
}

@InProceedings{	  efxsimplegraphs,
  author	= {George Christodoulou and Amos Fiat and Elias Koutsoupias
		  and Alkmini Sgouritsa},
  editor	= {Kevin Leyton{-}Brown and Jason D. Hartline and Larry
		  Samuelson},
  title		= {Fair allocation in graphs},
  booktitle	= {Proceedings of the 24th {ACM} Conference on Economics and
		  Computation, {EC} 2023, London, United Kingdom, July 9-12,
		  2023},
  pages		= {473--488},
  publisher	= {{ACM}},
  year		= {2023},
  url		= {https://doi.org/10.1145/3580507.3597764},
  doi		= {10.1145/3580507.3597764},
  bibsource	= {dblp computer science bibliography, https://dblp.org}
}

@Misc{		  feige2025multiallocationsallocationssubadditivevaluations,
  title		= {From multi-allocations to allocations, with subadditive
		  valuations},
  author	= {Uriel Feige},
  year		= {2025},
  eprint	= {2506.21493},
  archiveprefix	= {arXiv},
  primaryclass	= {cs.GT},
  url		= {https://arxiv.org/abs/2506.21493}
}

@Book{		  foley1966resource,
  title		= {Resource allocation and the public sector},
  author	= {Foley, Duncan Karl},
  year		= {1966},
  publisher	= {Yale University}
}

@Book{		  gamow1958puzzle,
  title		= {Puzzle-math},
  author	= {Gamow, G. and Stern, M.},
  isbn		= {9780670583355},
  lccn		= {58005402},
  url		= {https://books.google.gr/books?id=_vdytgAACAAJ},
  year		= {1958},
  publisher	= {Viking Press}
}

@Article{	  h__steihaus_1948,
  title		= { The problem of fair division },
  author	= { Hugo Steinhaus },
  journal	= { Econometrica },
  year		= { 1948 },
  volume	= { 16 },
  pages		= { 101-104 }
}

@Misc{		  hsu2024efxorientationsmultigraphs,
  title		= {EFX Orientations of Multigraphs},
  author	= {Kevin Hsu},
  year		= {2024},
  eprint	= {2410.12039},
  archiveprefix	= {arXiv},
  primaryclass	= {cs.GT},
  url		= {https://arxiv.org/abs/2410.12039}
}

@InProceedings{	  jahan2023rainbowcyclenumberefx,
  author	= {Shayan Chashm Jahan and Masoud Seddighin and Seyed
		  Mohammad Seyed Javadi and Mohammad Sharifi},
  title		= {Rainbow Cycle Number and {EFX} Allocations: (Almost)
		  Closing the Gap},
  booktitle	= {Proceedings of the Thirty-Second International Joint
		  Conference on Artificial Intelligence, {IJCAI} 2023,
		  19th-25th August 2023, Macao, SAR, China},
  pages		= {2572--2580},
  publisher	= {ijcai.org},
  year		= {2023},
  url		= {https://doi.org/10.24963/ijcai.2023/286},
  doi		= {10.24963/IJCAI.2023/286},
  bibsource	= {dblp computer science bibliography, https://dblp.org}
}

@inproceedings{kaviani2024envyfreeallocationindivisiblegoods,
  author       = {Alireza Kaviani and
                  Masoud Seddighin and
                  AmirMohammad Shahrezaei},
  title        = {Almost Envy-Free Allocation of Indivisible Goods: {A} Tale of Two
                  Valuations},
  booktitle    = {{WINE}},
  series       = {Lecture Notes in Computer Science},
  volume       = {15534},
  pages        = {261--276},
  publisher    = {Springer},
  year         = {2024}
}

@Misc{		  kaviani2025improvedapproximateefxguarantees,
  title		= {Improved Approximate EFX Guarantees for Multigraphs},
  author	= {Alireza Kaviani and Alireza Keshavarz and Masoud Seddighin
		  and AmirMohammad Shahrezaei},
  year		= {2025},
  eprint	= {2506.09288},
  archiveprefix	= {arXiv},
  primaryclass	= {cs.GT},
  url		= {https://arxiv.org/abs/2506.09288}
}

@inproceedings{Liptonetal,
  author       = {Richard J. Lipton and
                  Evangelos Markakis and
                  Elchanan Mossel and
                  Amin Saberi},
  title        = {On approximately fair allocations of indivisible goods},
  booktitle    = {{EC}},
  pages        = {125--131},
  publisher    = {{ACM}},
  year         = {2004}
}

@inproceedings{mixedmanna,
  author       = {Yu Zhou and
                  Tianze Wei and
                  Minming Li and
                  Bo Li},
  title        = {A Complete Landscape of {EFX} Allocations on Graphs: Goods, Chores
                  and Mixed Manna},
  booktitle    = {{IJCAI}},
  pages        = {3049--3056},
  year         = {2024}
}

@inproceedings{christodoulou2025exactapproximatemaximinshare,
  author       = {George Christodoulou and
                  Symeon Mastrakoulis},
  editor       = {Sven Koenig and
                  Chad Jenkins and
                  Matthew E. Taylor},
  title        = {Exact and Approximate Maximin Share Allocations in Multi-Graphs},
  booktitle    = {Fortieth {AAAI} Conference on Artificial Intelligence, Thirty-Eighth
                  Conference on Innovative Applications of Artificial Intelligence,
                  Sixteenth Symposium on Educational Advances in Artificial Intelligence,
                  {AAAI} 2026, Singapore, January 20-27, 2026},
  pages        = {16761--16769},
  publisher    = {{AAAI} Press},
  year         = {2026},
  url          = {https://doi.org/10.1609/aaai.v40i20.38719},
  doi          = {10.1609/AAAI.V40I20.38719},
  bibsource    = {dblp computer science bibliography, https://dblp.org}
}

@InProceedings{	  ontheexistenceofefxallocationsinmultigraphs,
  author	= {Sgouritsa, Alkmini and Sotiriou, Minas Marios},
  title		= {On the Existence of EFX Allocations in Multigraphs},
  year		= {2025},
  isbn		= {9798400714269},
  publisher	= {International Foundation for Autonomous Agents and
		  Multiagent Systems},
  booktitle	= {Proceedings of the 24th International Conference on
		  Autonomous Agents and Multiagent Systems},
  pages		= {2735–2737},
  numpages	= {3},
  location	= {Detroit, MI, USA},
  series	= {AAMAS '25}
}

@inproceedings{ZengStructureOrientationsGraphsAAMAS,
  author       = {Jinghan A. Zeng and
                  Ruta Mehta},
  title        = {On the Structure of {EFX} Orientations on Graphs},
  booktitle    = {{AAMAS}},
  pages        = {2309--2316},
  publisher    = {International Foundation for Autonomous Agents and Multiagent Systems
                  / {ACM}},
  year         = {2025}
}

@Article{	 plautrough,
  author	= {Benjamin Plaut and Tim Roughgarden},
  title		= {Almost Envy-Freeness with General Valuations},
  journal	= {{SIAM} J. Discret. Math.},
  volume	= {34},
  number	= {2},
  pages		= {1039--1068},
  year		= {2020},
  url		= {https://doi.org/10.1137/19M124397X},
  doi		= {10.1137/19M124397X},
  bibsource	= {dblp computer science bibliography, https://dblp.org}
}

@Article{	  procaccia,
  author	= {Ariel D. Procaccia},
  title		= {An answer to fair division's most enigmatic question:
		  technical perspective},
  journal	= {Commun. {ACM}},
  volume	= {63},
  number	= {4},
  pages		= {118},
  year		= {2020},
  url		= {https://doi.org/10.1145/3382131},
  doi		= {10.1145/3382131},
  bibsource	= {dblp computer science bibliography, https://dblp.org}
}

@Article{	  stromquist1980howtc,
  title		= {How to Cut a Cake Fairly},
  author	= {Walter R. Stromquist},
  journal	= {American Mathematical Monthly},
  year		= {1980},
  volume	= {87},
  pages		= {640-644},
  url		= {https://doi.org/10.1080/00029890.1980.11995109}
}

@Article{	  twovaluedinstanses,
  author	= {Georgios Amanatidis and Georgios Birmpas and Aris
		  Filos{-}Ratsikas and Alexandros Hollender and Alexandros A.
		  Voudouris},
  title		= {Maximum Nash welfare and other stories about {EFX}},
  journal	= {Theor. Comput. Sci.},
  volume	= {863},
  pages		= {69--85},
  year		= {2021},
  url		= {https://doi.org/10.1016/j.tcs.2021.02.020},

}

@Article{	  varian197463,
  title		= {Equity, envy, and efficiency},
  journal	= {Journal of Economic Theory},
  volume	= {9},
  number	= {1},
  pages		= {63-91},
  year		= {1974},
  issn		= {0022-0531},
  doi		= {https://doi.org/10.1016/0022-0531(74)90075-1},
  url		= {https://www.sciencedirect.com/science/article/pii/0022053174900751},
  author	= {Hal R Varian}
}

@Article{	  woo80,
  title		= {Dividing a cake fairly},
  journal	= {Journal of Mathematical Analysis and Applications},
  volume	= {78},
  number	= {1},
  pages		= {233-247},
  year		= {1980},
  issn		= {0022-247X},
  doi		= {https://doi.org/10.1016/0022-247X(80)90225-5},
  url		= {https://www.sciencedirect.com/science/article/pii/0022247X80902255},
  author	= {D.R Woodall}
}

\appendix

\section{Subroutine \ref{Algo:FindingMinimalEnviedSubset} (Algorithm 3 of \cite{CKMS21})}
\label{app:subroutine7}

\begin{algorithm}
\floatname{algorithm}{Subroutine}
\caption{Finding an inclusion-wise minimal envied subset from a patch (Algorithm 3 \cite{CKMS21})}
\label{Algo:FindingMinimalEnviedSubset}
\raggedright\textbf{Input:} A set of agents $L$, and a set of goods $S$ \\
\textbf{Output:} A set that none envies a proper subset of it.\\
\begin{algorithmic}[1]
\STATE $Z=S$
\FOR{ every agent $i \in L$}
    \FOR{ every good $g \in Z$}
        \IF{ $v_i(X_i) < v_i(Z\setminus\{g\})$}
            \STATE $Z \gets Z\setminus \{g\}$
        \ENDIF
    \ENDFOR
\ENDFOR
\RETURN Z

\end{algorithmic}
\end{algorithm}

\end{document}